\documentclass[conference]{IEEEtran}
\usepackage[utf8]{inputenc}
\usepackage{amsmath}
\usepackage{tikz}
\usetikzlibrary{shapes,arrows}
\usepackage{amssymb}
\usepackage{latexsym}
\usepackage{ifpdf}
\newtheorem{theorem}{Theorem}
\newtheorem{lemma}{Lemma}

\newtheorem{corollary}{Corollary}

\begin{document}
\title{Transmission Rank Selection for Opportunistic Beamforming with Quality of Service Constraints}
\author{\IEEEauthorblockN{Meng Wang}
\IEEEauthorblockA{Department of Electrical and Electronic Engineering,\\
University of Melbourne, Australia\\
Email: meng.wang@unimelb.edu.au}\and \IEEEauthorblockN{Tharaka Samarasinghe, Jamie S. Evans}
\IEEEauthorblockA{Department of Electrical and Computer Systems Engineering,\\
Monash University, Australia\\
Email: \{tharaka.samarasinghe, jamie.evans\}@monash.edu}\vspace{-0.5cm}}

\maketitle
\begin{abstract}
In this paper, we consider a multi-cell multi-user MISO broadcast channel. The system operates according to the opportunistic beamforming framework in a multi-cell environment with variable number of transmit beams (may alternatively be referred as the transmission rank) at each base station. The maximum number of co-scheduled users in a cell is equal to its transmission rank, thus increasing it will have the effect of increasing the multiplexing gain. However, this will simultaneously increase the amount of interference in the network, which will decrease the rate of communication. This paper focuses on optimally setting the transmission rank at each base station such that a set of Quality of Service (QoS) constraints, that will ensure a guaranteed minimum rate per beam at each base station, is not violated. Expressions representing the achievable region of transmission ranks are obtained considering different network settings. The achievable transmission rank region consists of all achievable transmission rank tuples that satisfy the QoS constraints. Numerical results are also presented to provide further insights on the feasibility problem. 
\end{abstract}

\section{Introduction}
Opportunistic beamforming (OBF) is a well known adaptive signaling scheme that has received a great deal of attention in the literature as it attains the sum-rate capacity with full channel state information (CSI) to a first order for large numbers of mobile users in the network, while operating on partial CSI feedback from the users. In this paper, we consider a cellular network which operates according to the OBF framework in a multi-cell environment with variable number of transmit beams at each BS. The number of transmit beams is also referred to as the transmission rank (TR) in the paper, and we focus on optimally setting the transmission rank at each BS in the network.

The earliest work of OBF appeared in the landmark paper~\cite{Tse2002}, where the authors have introduced a single-beam OBF scheme for the single-cell multiple-input single-output (MISO) broadcast channel. The concept was extended to $N_t$ random orthogonal beams in~\cite{Hassibi2005}, where $N_t$ is the number of transmit antennas. The downlink sum-rate of this scheme scales as $N_t\log\log(K)$, where $K$ is the number of users in the system~\cite{Hassibi2005}. Recently, the authors in \cite{vtc2013} have considered using variable TRs at the BS, and they have showed that the downlink sum-rate scales as $\frac{L}{L-1}\log(K)$ in interference-limited networks, where $L$ is the TR employed by the BS. The gains of adapting variable TR compared to a fixed one is clearly demonstrated in~\cite{vtc2013}, however, how to select the TR for OBF is still an open question which has only been characterized in the asymptotic sense for the single-cell system in~\cite{Wagner2006}-\cite{Vicario2008}, and a two-transmit antenna single-cell system in~\cite{Rocky2013}. 

In all of the above works, the users are assumed to be homogeneous with the large-scale fading gain (alternatively referred to as the path loss in this paper) equal to unity. OBF in heterogenous networks has been considered in~\cite{Huang2012}-\cite{ausctw2014}. In~\cite{Huang2012}, the authors focused on the fairness of the network and obtained an expression for the ergodic capacity of this fair network. In \cite{Tharaka2013}, the authors modeled the user locations using a spatial Poisson point process, and studied the outage capacity of the system. In~\cite{ausctw2014}, the authors considered an interference-limited network and derived the ergodic downlink sum-rate scaling law of this interference-limited network. The TRs in~\cite{Huang2012}-\cite{ausctw2014} are considered to be fixed. 

In this paper, we are interested in the Quality of Service (QoS) delivered to the users. More precisely, we focus on a set of QoS constraints that will ensure a guaranteed minimum rate per beam with a certain probability at each BS. Previous studies have shown that user's satisfaction is a non-decreasing, concave function of the service rate~\cite{Enderle2003}; this suggests that the user's satisfaction is insignificantly increased by a service rate higher than what the user demands, but drastically decreased if the provided rate is below the requirement~\cite{Zorba2008}. The network operator can promise a certain level of QoS to a subscribed user. To this end, the QoS is closely related to the TR of the BS. Increasing the TR will increase the number of co-scheduled users. However, increasing the TR will also increase interference levels in the network, which will decrease the rate of communication per beam. A practical question arises; what is a suitable TR to employ at each BS while achieving a certain level of QoS in multi-cell heterogeneous networks? The authors in~\cite{Zorba2007} have performed a preliminary study of this problem for a single-cell system consisting of homogeneous users with identical path loss values of unity. 

The main contributions of this paper are summarized as follows. We focus on finding the achievable TRs without violating the above mentioned set of QoS constraints. This can be formulated into a feasibility problem. For some specific cases, we derive analytical expressions of the achievable TR region, and for the more general cases, we derive expressions that can be easily used to find the achievable TR region. The achievable TR region consists of all the achievable TR tuples that satisfy the QoS constraints. Numerical results are presented for a two-cells scenario to provide further insights on the feasibility problem; our results show that the achievable TR region expands when the QoS constraints are relaxed, the SNR and the number of users in a cell are increased, and the size of the cells are decreased. 

\section{System Model}\label{sec:sys_model}
We consider a multi-cell multi-user MISO broadcast channel. The system consists of $M$ BSs (or cells), each equipped with $N_t$ transmit antennas. Each cell consists of $K$ users, each equipped with a single receive antenna. A BS will only communicate with users in its own cell. Let $\mathbf{h}_{j,i,k}$ denote the channel gain vector between BS $j$ and user $k$ in cell $i$. The elements in $\mathbf{h}_{j,i,k}$ are independent and identically distributed (i.i.d.) random variables, each of which is drawn from a zero mean and unit variance {\em circularly-symmetric complex Gaussian} distribution $\mathcal{CN}(0,1)$. The large-scale fading gain (may alternatively referred to as path loss in this paper) between BS $j$ and user $k$ in cell $i$ is denoted by $g_{j,i,k}$. The path loss (PL) values of all the users are governed by the PL model $G(d)=d^{-\alpha}$ for $\alpha>2$, where $d$ represents the distance between the user and the BS of interest. Therefore, the random PL values are also i.i.d. among the users, where the randomness stems from the fact that users' locations are random. Moreover, we assume a quasi-static block fading model over time~\cite{tse_book}. 

The BSs operate according to the OBF scheduling and transmission scheme as follows. The BSs will first pre-determine the number of beams to be transmitted. BS $i$ generates $L_i$ random orthogonal beamforming vectors and transmits $L_i$ different symbols along the direction of these beams ($L_i$ is the TR employed by BS $i$). This process is simultaneously carried out at all BSs. For BS $i$, let $\mathbf{w}_{i,m}$ and ${s}_{i,m}$ denote the beamforming vector and the transmitted symbol on beam $m$, respectively. The received signal at user $k$ in cell $i$ can be written as

\vspace{-0.5cm}
{\small
\begin{multline}
 y_{i,k} =  \sqrt{g_{i,i,k}}\  \mathbf{h}_{i,i,k}^\top \sum_{l=1}^{L_i} \mathbf{w}_{i,l} {s}_{i,l} +  \\
 \sum_{j\neq i}^M \sqrt{g_{j,i,k}}\  \mathbf{h}_{j,i,k}^\top \sum_{t=1}^{L_j} \mathbf{w}_{j,t} {s}_{j,t} + {n}_{i,k},
\end{multline}}%
where ${n}_{i,k} \sim \mathcal{CN}(0,\sigma_n^2)$ is the additive complex Gaussian noise. We assume that $\mathbf{E}[|{s}_{i,k}|^2] = \rho_i$, where $\rho_i = P_t/L_i$ is a scaling parameter to satisfy the total power constraint $P_t$ at each BS. For conciseness, we assume $P_t=1$. Each user will measure the SINR values on the beams from its associated BS, and feed them back. For the beam generated using $\mathbf{w}_{i,m}$, the received SINR at user $k$ located in cell $i$ is given by
{\small
\begin{align} \label{eq:SINR_expression}
S_{i,k,m} &= \left[g_{i,i,k}|\mathbf{h}_{i,i,k}^\top \mathbf{w}_{i,m}|^2\right] \left\{\sigma_n^2 L_i + g_{i,i,k} \sum_{\substack{l\neq m \\ l=1}}^{L_i} |\mathbf{h}_{i,i,k}^\top \mathbf{w}_{i,l}|^2 \right. \nonumber \\ & \hspace{1.5cm} \left.
+ \sum_{j \neq i}^M g_{j,i,k} \frac{L_i}{L_j} \sum_{\substack{t=1}}^{L_j} |\mathbf{h}_{j,i,k}^\top \mathbf{w}_{j,t}|^2 \right\}^{-1}.
\end{align}}

\vspace{-0.5cm}
Once the BSs have received the feedback from the users, each BS will select a set of users for communication by assigning each beam to the in-cell user having the highest SINR on it\footnote{In this paper, we focus on rate maximization in the network. Interested readers are referred to \cite{Tse2002}, \cite{Rocky2013},  \cite{TharakaPE} and \cite{Huang2013} for techniques that can be used to achieve fairness in such a network.}, \emph{i.e.}, the user with SINR value $S^\star_{i,m} = \max_{1\leq k\leq K} S_{i,k,m}$. For cell $i$, let $F_{L_i}^k$ and $F^\star_{L_i}$ denote the distributions of the SINR on a beam at user $k$ and the maximum SINR on a beam, respectively. 
Since the maximum number of co-scheduled users in cell $i$ is equal to $L_i$, increasing $L_i$ will have the effect of increasing the number of co-scheduled users. However, increasing $L_i$ will also increase the amount of intra-cell and inter-cell interference, which will decrease the rate of communication per beam. 
Therefore, we focus on finding an achievable $(L_1,\ldots,L_M)$ TR M-tuple with a set of QoS constraints at all the BSs that will ensure a guaranteed minimum rate per beam with a certain probability. To this end, we consider that an outage probability of $p$ can be tolerated at each BS, where the outage event refers to the received SINR of the scheduled user on a beam being below a target SINR threshold value $\eta$, \emph{i.e.}, $\Pr\{ S^\star_{i,m} \leq \eta \} \leq p \Rightarrow F^\star_{L_i}(\eta)\leq p$ for all $i$. There is also a natural constraint on $L_i$ due to the orthogonality requirement among the beams, \emph{i.e.}, $L_i \leq N_t$ for all $i$. We focus on finding the achievable $L_i$s such that these constraints are not violated. This is a non-trivial problem for the system of interest due to the presence of intra-cell and inter-cell interferences, and the SINR values on a beam being not identically distributed among the users due to their different locations. 
We note that there is an implicit constraint that $L_i$ must be an integer. For the analysis, we will relax the integer constraint and assume that $N_t$ is sufficiently large such that the constraints $0 \leq L_i \leq N_t$ is always satisfied for all $i$. Denote $(\tilde{L}_1,\ldots,\tilde{L}_M)$ as an achievable TR M-tuple with the relaxed constraints; the corresponding achievable $(L_1,\ldots,L_M)$ M-tuple is given by $L_{i} = \min(N_t, \lfloor \tilde{L}_{i}\rfloor)$ for all $i$, where $\lfloor.\rfloor$ represents the floor function. 
Since the SINR on a beam is a strictly decreasing function of the TR, we have the following property; given an achievable M-tuple $(\hat{L}_1,\ldots,\hat{L}_M)$, another M-tuple $(\bar{L}_1,\ldots,\bar{L}_M)$ is achievable if $\bar{L}_i \leq \hat{L}_i$ for all $i=1,\ldots,M$. In the remaining parts of the paper, we will focus on finding the achievable TRs and the achievable TR region, where the achievable TR region is defined to  consist all the achievable $(\tilde{L}_1,\ldots,\tilde{L}_M)$ M-tuples. We will call the constraints on $F^\star_{L_i}(\eta)$ the QoS constraints.

\section{Analysis for a Single Cell Scenario}\label{sec:L}
We will start our analysis with a simple single-cell scenario. We drop the cell index~$i$ for brevity. For a single cell, the SINR expression in (\ref{eq:SINR_expression}) reduces to
{\small
\begin{align}\label{eq:SINR_expression_single_cell}
S_{k,m} = \frac{g_k |\mathbf{h}_k^\top \mathbf{w}_m|^2}{\sigma_n^2 L+ g_k \sum_{\substack{l\neq m \\ l=1}}^{L} |\mathbf{h}_k^\top \mathbf{w}_l|^2}.
\end{align}}%
For a given PL value $g_k$, by using techniques similar to those used in \cite{Hassibi2005}, it is not hard to show that $F_{L}^k$ is given by
{\small
\begin{align}
F^k_{L}(x|g_k) = 1 - \frac{\exp\left(-x\sigma_n^2 L/g_k\right)}{(x+1)^{L-1}}.
\end{align}}\hspace{-0.1cm}
Therefore, by conditioning on $\mathbf{g}=\{g_1,\ldots,g_K\}$, the CDF of $S_{m}^\star$ is given by
{\small
\begin{align} \label{eq:cdf_sc_noniid}
F_{L}^{\star}(x|\mathbf{g})= \prod_{k=1}^K \left[1 -  \frac{\exp\left(-x\sigma_n^2 L/g_k\right)}{(x+1)^{L-1}}\right].
\end{align}}

\subsection{Homogeneous Users with Identical PL Values}
First we consider the simplest case where the user's are located equidistant to the BS, \emph{i.e.}, the user's PL values are identical and deterministic, and given by $g_1 = \ldots = g_K = g$. For this simplest case, a closed-form expression for the achievable TR can be obtained, and it is formally presented through the following theorem.
\begin{theorem}\label{thm:single_homo_iid}
For the system in consideration with $M=1$ and $g_1=\ldots=g_K=g$, the achievable TRs are given by 
{\small
\begin{align}\label{eq:feasible_soln_iid2}
\tilde{L} \leq \frac{\log(1+\eta) - \log(1-p^{1/K})}{\eta\sigma_n^2/g + \log(1+\eta)},
\end{align}}%
where $\eta$ is the target SINR threshold value.
\end{theorem}
\begin{IEEEproof}
With equal PL values $g_1 = \ldots = g_K = g$, the QoS constraint is given by
{\small
\begin{align*}
\left[1 - \frac{\exp\left(-\eta\sigma_n^2 L/g\right)}{(\eta+1)^{L-1}} \right]^K \leq p.
\end{align*}}%
Solving for $L$ completes the proof. 
\end{IEEEproof}
Setting $g=1$ makes the result in Theorem~\ref{thm:single_homo_iid} consistent with~\cite{Zorba2007}.

\subsection{Heterogeneous Users with Random PL Values}
Now, we will consider the users to be heterogenous as in Section \ref{sec:sys_model}. We model the cell as a disk with radius $D$. Given the non-identical PL values, $F_L^{\star}$ is given by (\ref{eq:cdf_sc_noniid}) for this setup, and the QoS constraint can be written as
\begin{align}\label{eq:single_cell_convexity}
\prod_{k=1}^{K} \left[1 - \frac{\exp(-(\eta\sigma_n^2/g_k)L)}{(\eta+1)^{L-1}} \right] \leq p.
\end{align}
Since the user locations are random in our setup, removing the conditioning of $F_L^\star$ by averaging over the PL values gives us the QoS constraint of interest. This idea is formally presented in the following lemma.
\begin{lemma} \label{lem:unbounded_path_model}
For the system in consideration with $M=1$ and the random PL values governed by the PL model $G(d)=d^{-\alpha}$ for $\alpha>2$, the QoS constraint is given by
{\small
\begin{align}\label{eq:qos_constraint_unbounded_pathloss_model}
\left[1 - \frac{2\gamma(\frac{2}{\alpha}, \eta \sigma_n^2L D^\alpha)}{\alpha D^2 (\eta+1)^{L-1} (\eta \sigma_n^2L)^{2/\alpha}}  \right]^K \leq p,
\end{align}}%
where $\eta$ is the target SINR threshold value and $\gamma(\cdot,\cdot)$ is the lower incomplete gamma function. 
\end{lemma}
\begin{IEEEproof}
Since the users are located uniformly over the plane, the CDF of the distance from the user to its associated BS is given by $\Phi_D(d)=(d/D)^2$. Let $\Phi_G(g)$ denote the CDF of the PL value, which is given by $\Phi_G(g) = 1 - \frac{g^{-2/\alpha}}{D^2}$. Since the PL values are i.i.d. among the users, we have
{\small
\begin{align*}
F_{L}^{\star}(x) &= \int_{g_1=G(D)}^{G(0)} \cdots \int_{g_K=G(D)}^{G(0)} \prod_{k=1}^K F^k_{L}(x|g_k) d\Phi_G(g_1) \ldots d\Phi_G(g_K)\\
&=\left[\int_{G(D)}^{G(0)}  F_{L}(x|g) d\Phi_G(g)\right]^K.
\end{align*}}%
Substituting for the CDFs and setting $t = g^{-2/\alpha}$ gives us
{\small
\begin{eqnarray*}
F_{L}^{\star}(\eta) = \left[1 - \frac{1}{D^2 (\eta+1)^{L-1}} \int_{0}^{D^2} \exp\left(-\eta\sigma_n^2L t^{\alpha/2}\right) dt \right]^K.
\end{eqnarray*}}%
Evaluating the integral completes the proof \cite{table_of_integral}.
\end{IEEEproof}
The achievable TR region consists of all the achievable $\tilde{L}$s that satisfy (\ref{eq:qos_constraint_unbounded_pathloss_model}). Next, we will focus on the general multi-cell scenario.

\section{Analysis for the Multi-Cell Scenario} \label{sec:L2}
Similar to what we have done in the previous section, we will start the analysis by obtaining an expression for the conditional distribution of the SINR on a beam at a user. This result is formally presented in the following lemma.
\begin{lemma}\label{lem:sinr_cdf}
Consider user $k$ in cell $i$. Given the PL values from all the BSs to user $k$, {\em i.e.}, given $\mathbf{g}_{i,k} = \{g_{1,i,k},\ldots,g_{M,i,k}\}$, the conditional distribution of the SINR on a beam is given by
{\small
\begin{align}
F^k_{L_i}(x|\mathbf{g}_{i,k}) = 1 - \frac{\exp\left(-\frac{x{\sigma}_n^2 L_i}{g_{i,i,k}}\right)} {(x+1)^{L_i-1}\prod_{\substack{j\neq i \\ j=1}}^{M} \left(x \frac{g_{j,i,k}}{g_{i,i,k}} \frac{L_i}{L_j} +1 \right)^{L_j}}. \label{eq:sinr_distribution}
\end{align}}%
\end{lemma}
\begin{IEEEproof}
The conditional distribution of the SINR can be obtained using a result in \cite{Boyd2002}, which is summarized as follows. Suppose $Z_i$, $i=1,\ldots,n$ are independent exponentially distributed random variables (RVs) with parameters $\lambda_i$. Then
{\small
\begin{align*}
\Pr\left(Z_1 \leq c+\sum_{i=2}^{n} Z_i \right) = 1 - \exp(-\lambda_1 c)\prod_{i=2}^{n} \left(\frac{1}{1+\frac{\lambda_1}{\lambda_i}}  \right),
\end{align*}}\hspace{-0.15cm}
where $c$ is a constant. Given all the PL values $\mathbf{g}_{i,k}$, $F_{L_i}^k$ can be re-written as
{\small\begin{align*}
F_{L_i}^k(x|\mathbf{g}_{i,k}) &= \Pr \left\{S_{i,k,m} \leq x |\mathbf{g}_{i,k} \right\} \\
		   &= \Pr \left\{Z_{i,m} \leq c + \sum_{l\neq m}^{L_i} A_{i,l} + \sum_{j\neq i}^{M} \sum_{t=1}^{L_j} B_{j,t}  \right\}, 
\end{align*}}%
where $c= x \sigma_n^2 L_i/g_{i,i,k}$ is a constant, and $Z_{i,m}, A_{i,l}$ and $B_{j,t}$ are independent exponentially distributed RVs with parameters $1, \frac{1}{x}$ and $\frac{g_{i,i,k}L_j}{x g_{j,i,k}L_i}$, respectively. Therefore, directly using the result in \cite{Boyd2002} completes the proof. 
\end{IEEEproof}

Using the above lemma, given all PL values, the conditional CDF of the maximum SINR on a beam can be written as
{\small
\begin{align}\label{eq:cdf_mc_noniid}
F_{L_i}^\star (x|\mathbf{g}_{1,i},\ldots,\mathbf{g}_{K,i}) = \prod_{k=1}^{K} F_{L_i}^k(x|\mathbf{g}_{i,k}).
\end{align}}%
Next, we will use this expression to find the achievable TR region considering different scenarios, similar to what we have done in Section \ref{sec:L}. For the clarity of presentation and the ease of explanation, we present the analysis for the two-cells scenario; the analysis of the $M$-cells scenario can be easily extended using the same techniques. 

\subsection{The Wyner Model}
First we consider the classical Wyner model~\cite{Wyner_model} for the two-cells scenario. The users' PL values are deterministic as follows; the PL value between all the users to their associated BS is unity, and the PL value between all the users to the interfering BS is $g$. For this setup, the QoS constraint for cell one is given by
{\small
\begin{align}\label{eq:QoS_mc_IID_Unequal}
\left[1 -  \frac{\exp\left(-\eta \sigma_n^2 L_1\right)}{(\eta +1)^{L_1-1} \left(\frac{L_1}{L_2}g\eta +1 \right)^{L_2}}\right]^K \leq p,
\end{align}}%
where $\eta$ is the target SINR threshold. The QoS constraint for cell two can be easily obtained by interchanging $1$ and $2$ in the indices. Analytical expressions that characterize the achievable TR region for this setup are formally presented through the following theorem. 
\begin{theorem} \label{thm:wyner_model}
For the Wyner model, given a fixed $L_2$, the achievable TRs for cell one is given by
{\small
\begin{align}\label{eq:wyner_model_closed_form_soln}
\tilde{L}_1 \leq -\frac{1}{c}+\frac{b}{a} \mathcal{W} \left(\frac{a}{bc}+\frac{d}{b} \right),
\end{align}}%
where $\mathcal{W}$ is the Lambert-W function given by the defining equation $\mathcal{W}(x)\exp(\mathcal{W}(x))=x$, $a=\eta \sigma_n^2+\log(1+\eta)$, $b=L_2$, $c = g \eta/L_2$, $d = \log(1+\eta)-\log(1-p^{1/K})$, and $\eta$ is the target SINR threshold. 
\end{theorem}
\begin{proof}
With some simple manipulations, we can re-write the QoS constraint in (\ref{eq:QoS_mc_IID_Unequal})
{\small
\begin{align}\label{eq:two_cell_simplified_closed_form_qos}
aL_1 + b\log(1+cL_1) \leq d.
\end{align}}%
The following chain of inequalities holds which completes the proof. 
{\small
\begin{align*}
& aL_1 + b \log(1+cL_1) \leq d \\
&\Rightarrow \frac{aL_1}{b} + \log(1+cL_1) + \frac{a}{bc} \leq \frac{d}{b} + \frac{a}{bc} \\
&\Rightarrow \exp\left(\frac{aL_1}{b}+\frac{a}{bc}\right) \left(\frac{aL_1}{b}+\frac{a}{bc}\right) \leq \frac{a}{bc}\exp\left(\frac{d}{b}+\frac{a}{bc}\right) \\
&\Rightarrow  L_1 \leq -\frac{1}{c} + \frac{b}{a}\mathcal{W}\left(\frac{a}{bc}+\frac{d}{b}\right).
\end{align*}}%
\end{proof}
Given a fixed $L_1$, the achievable TRs for cell two can be easily obtained by interchanging $1$ and $2$ in the indices. The achievable TR region consists of all the achievable $(\tilde{L}_1,\tilde{L}_2)$ tuples. When $L_1 = L_2 = L$, the result in Theorem~\ref{thm:wyner_model} can be further simplified, and the result is presented in the following corollary. 
\begin{corollary}\label{cor:Wyner_model}
For the Wyner model, if $L_1 = L_2 = L$, the achievable TRs are given by
{\small
\begin{align}\label{eq:feasible_soln_iid1}
\tilde{L} \leq \frac{\log(1+\eta) - \log(1-p^{1/K})}{\eta\sigma_n^2 + \log(1+\eta) + \log(1+g\eta)}.
\end{align}}%
\end{corollary}
Next, we consider the users to be heterogeneous as in Section~\ref{sec:sys_model}. 

\subsection{Heterogeneous Users with Random PL Values}
For this scenario, if all the path loss values are given, the QoS constraint for cell one can be written using \eqref{eq:cdf_mc_noniid} as
{\small
\begin{align}\label{eq:noniid_constraint}
\prod_{k=1}^{K} \left[1 - \frac{\exp\left(-\eta\sigma_n^2 L_1/g_{1,1,k}\right)}{(\eta+1)^{L_1-1}  \left(\frac{L_1}{L_2} \frac{g_{2,1,k}}{g_{1,1,k}} \eta +1 \right)^{L_2}}\right] \leq p.
\end{align}}%
The QoS constraint for cell two can be easily obtained by interchanging $1$ and $2$ in the indices. Since the user locations are random, we need to remove the conditioning on $F_{L_i}^\star$ by averaging over the PL values. With multiple BSs, the PL values between the user and each BS are correlated. Hence, it is difficult to average over the PL values directly as in the single cell case because it is difficult to obtain the CDF of the path loss value. Nonetheless, since the PL values are directly related to the distance between the user and each of the BSs, we can perform a change of variables by writing each PL as a function of the user and BS locations, and then average over the location process by making use of the fact that the users are located uniformly over the plane. 

For the purpose of illustrating the idea, consider a user $k$ in cell one and let $(x_{1,k},y_{1,k})$ denote its exact location coordinate on the two dimensional plane. For convenience, we assume that a user is always connected to the closest BS geographically, \emph{i.e.}, the two cells are arranged in a rectangular grid on the two dimensional plane. Hence $x_{1,k}$ and $y_{1,k}$ are independent and uniformly distributed within the cell for all $k$. Let $(X_i,0)$ denote the location coordinate of BS $i$. Figure~\ref{fig:two_cell_model} illustrates the setup. 
\begin{figure}[ht!]
	\centering
	\vspace{-0.25cm}
		\includegraphics[width=3.2in]{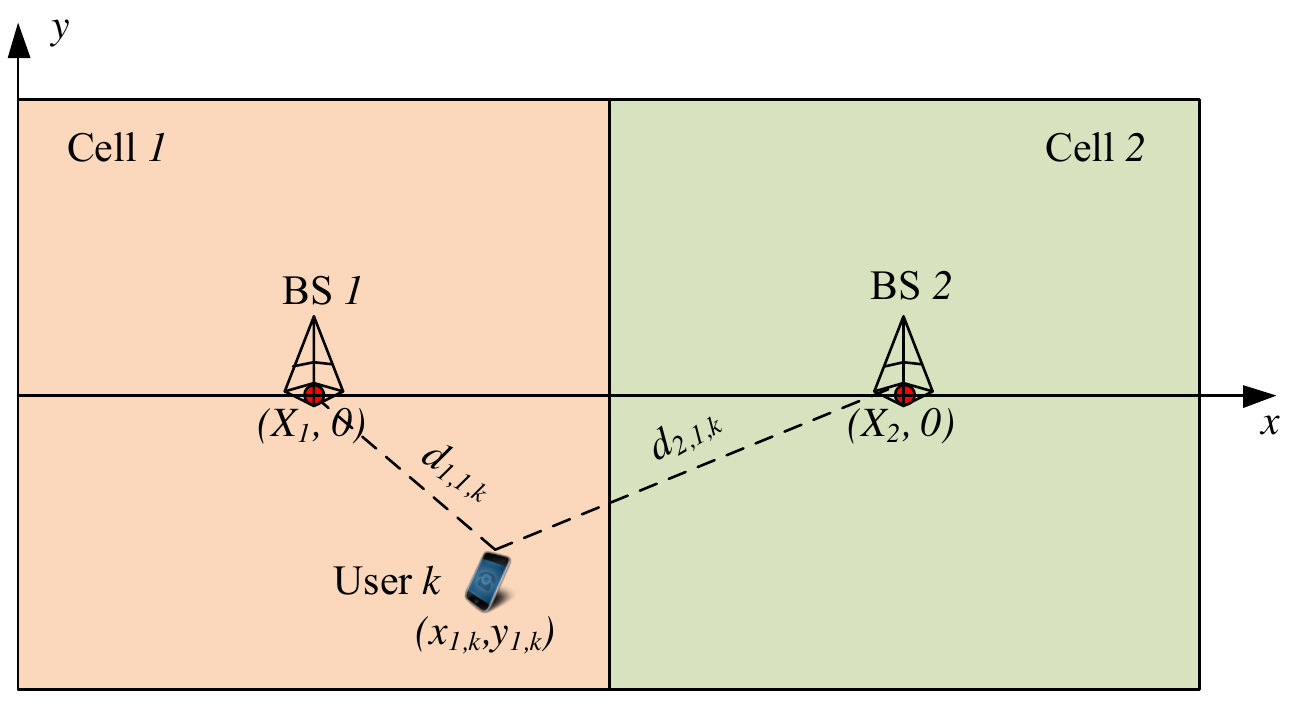}
	\caption{Two Cells Model}
	\label{fig:two_cell_model}
\end{figure}
\vspace{-0.25cm}
The distance from the user to BS one and two is therefore $d_{1,1,k}=\sqrt{(x_{1,k}-X_1)^2+(y_{1,k})^2}$ and $d_{2,1,k}=\sqrt{(x_{1,k}-X_2)^2+(y_{1,k})^2}$, respectively. Thus the PL values are given by $g_{1,1,k} =d_{1,1,k}^{-\alpha}$ and $g_{2,1,k}=d_{2,1,k}^{-\alpha}$, respectively. 
The following lemma presents the QoS constraints for a two-cells scenario consisting of heterogeneous users with random PL values. 
\begin{lemma}\label{lem:multi_cell_hetero}
For the system in consideration with BS $i$ being located at $(X_i,0)$, given a fixed $L_2$, the QoS constraint of cell one is  
\begin{align}\label{eq:QoS_multi_cell_last}
\frac{\Omega_{L_1}(\eta)}{\mathcal{A}_1 (\eta+1)^{L_1-1}} \geq 1-p^{1/K}, 
\end{align}
where $\eta$ is the target SINR, $\mathcal{A}_1$ is the area of cell one, and $\Omega_{L_1}(\eta)$ is defined by the following integral
{\small
\begin{align}
\Omega_{L_1}(\eta) = \int_{y} \int_{x} \frac{\exp\left(-\eta\sigma_n^2 L_1[(x-X_1)^2+y^2]^{\alpha/2}\right)}{\left[ \left(\frac{(x-X_1)^2+y^2}{(x-X_2)^2+y^2}\right)^{\frac{\alpha}{2}} \frac{L_1}{L_2}\eta+1\right]^{L_2}} dx dy,
\end{align}}%
and the integration is over the area of cell one. 
\end{lemma}
\begin{IEEEproof}
First we substitute $g_{1,1,k}$ and $g_{2,1,k}$ to (\ref{eq:sinr_distribution}) to get $F^k_{L_1}$. Given a user's location coordinate $(x_{1,k},y_{1,k})$, $F^k_{L_1}$ is given by
{\small
\begin{align*}
&F^k_{L_1}(s|x_{1,k},y_{1,k}) =  1 - \exp\left(-\frac{s\sigma_n^2 L_1}{[(x_{1,k}-X_1)^2+(y_{1,k})^2]^{-\alpha/2}} \right) \\ & \left[(s+1)^{L_1-1} \left( \left(\frac{(x_{1,k}-X_1)^2+(y_{1,k})^2}{(x_{1,k}-X_2)^2+(y_{1,k})^2}\right)^{\frac{\alpha}{2}} \frac{L_1}{L_2}s+1\right)^{L_2}\right]^{-1}.
\end{align*}}%
Averaging (\ref{eq:cdf_mc_noniid}) over cell one gives us
{\small
\begin{multline}
F_{L_1}^\star(\eta) = \int_{x_{1,K}} \int_{y_{1,K}} \cdots \int_{x_{1,1}} \int_{y_{1,1}} \prod_{k=1}^{K} F_{L_1}^k(\eta|x_{1,k},y_{1,k}) \\ f(x_{1,k},y_{1,k}) dx_{1,k} dy_{1,k}, \nonumber
\end{multline}}%
where $f(x_{1,k},y_{1,k}) = f(x_{1,k})f(y_{1,k}) = \frac{1}{\mathcal{A}_1}$ is the joint PDF of $x_{1,k}$ and $y_{1,k}$. Since the location coordinates are i.i.d. among the users, we have
{\small
\begin{align*} \label{eq:mc_ub_given_K}
F_{L_1}^\star(\eta)&=\left[ \int_{y} \int_{x} F^k_{L_1}(\eta|x,y) f(x,y) dx dy \right]^K.
\end{align*}}%
Substituting for $F^k_{L_1}(\eta|x,y)$ completes the proof.
\end{IEEEproof}
Given a fixed $L_1$, the QoS constraint for cell two can be easily obtained by interchanging $1$ and $2$ in the indices. The achievable TR region is given by all the $(\tilde{L}_1,\tilde{L}_2)$ tuples that satisfy the QoS constraints for both cells.

\section{Numerical Results}\label{sec:results}
In this section, we present our numerical results for the single cell and two-cell scenarios. In all the simulations, the cell is modeled as a disk with radius $D$ for the single cell scenario, and each cell is modeled as a square with cell area $\mathcal{A}_1=\mathcal{A}_2=(2D)^2$ for the two-cell scenario. 

In Figures~\ref{fig:region_homo1}-\ref{fig:region_hetero}, we show the achievable TR regions for the two-cells scenario. Figures~\ref{fig:region_homo1} and~\ref{fig:region_homo2} show the achievable TR regions for the Wyner model with $g=0.1$ and $g=1$, respectively. The dotted line connecting the origin and the corner point in each region represents the achievable TR set given in Corollary~\ref{cor:Wyner_model}.  Figure~\ref{fig:region_hetero} shows the achievable TR regions for heterogeneous users, using the result of Lemma~\ref{lem:multi_cell_hetero}. For a given ${L}_1$, any ${L}_2$ below the boundary can be achieved, whereas any ${L}_2$ above the boundary will violate the QoS constraints. Moreover, if the system wants to maximize the multiplexing gain at each BS, operating at $(L_1,L_2)$ strictly below the boundary is sub-optimal in a sense that we can further increase the TRs without violating the QoS constraints. Therefore, the boundary curve can be considered as the Pareto optimal boundary between the achievable and un-achievable TR pairs. As can be observed from the figures, TR region expands when the QoS constraints are relaxed, \emph{i.e.}, $p$ is increased and/or $\eta$ is decreased. Relaxing the QoS constraints allows more interference in the network, thus expanding the achievable TR Region. Moreover, the achievable TR region also expands when $K$ is increased. The achievable rate on a beam increases due to multi-user diversity, therefore, more beams/interference can be tolerated without violating the constraints. The achievable TR region will also change with $D$ and $\sigma_n^2$ and will be discussed further in Figures~\ref{fig:vsK}-\ref{fig:vsSNR}. 

\begin{figure}[ht!]
	\centering
	\includegraphics[width=2.7in]{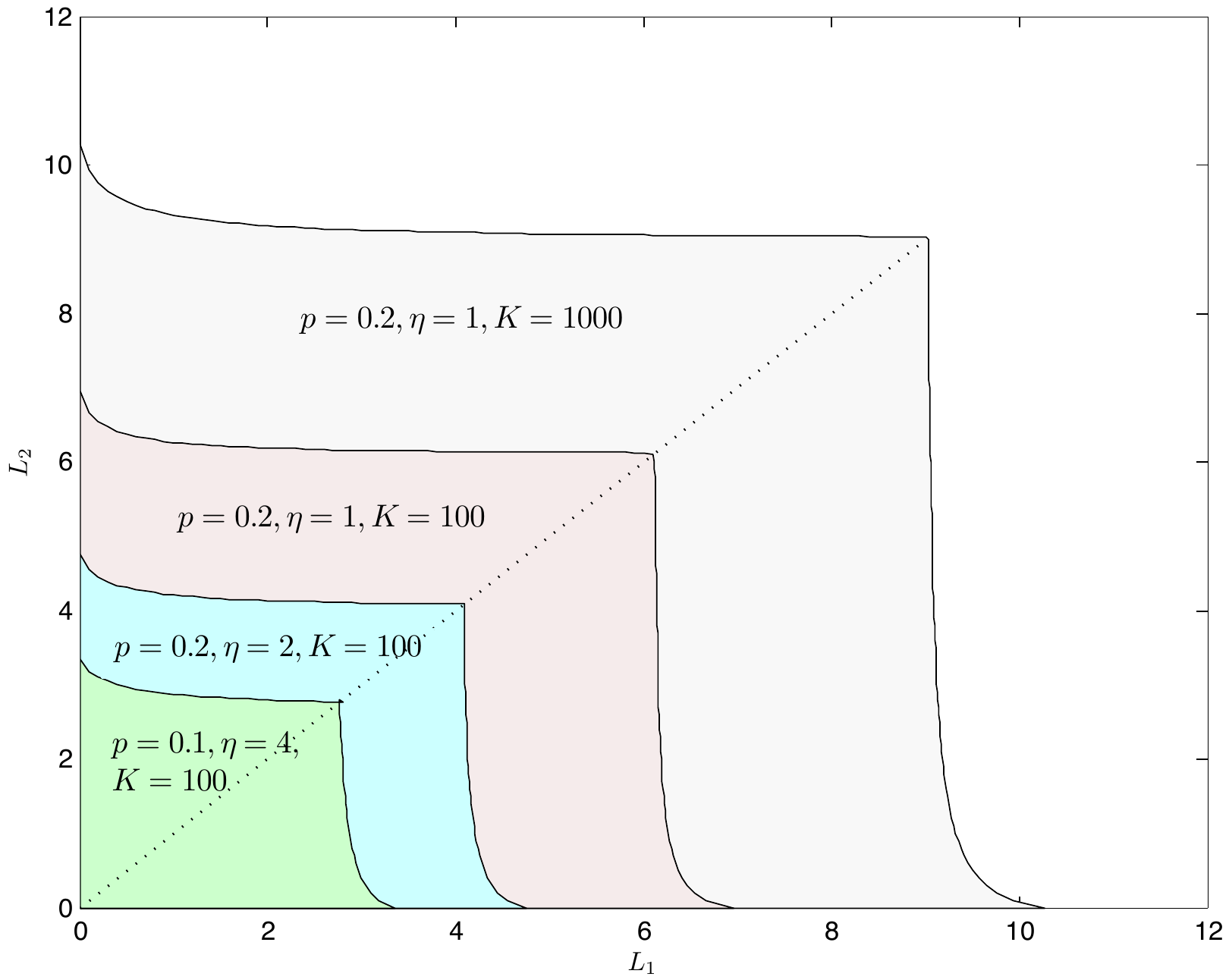}
	\caption{Achievable TR Region for the Two-Cells Scenario with Wyner Model, $D=2, \sigma_n^2=0.01, g=0.1$.}
	\label{fig:region_homo1}
\end{figure}

\begin{figure}[ht!]
	\centering
	\includegraphics[width=2.7in]{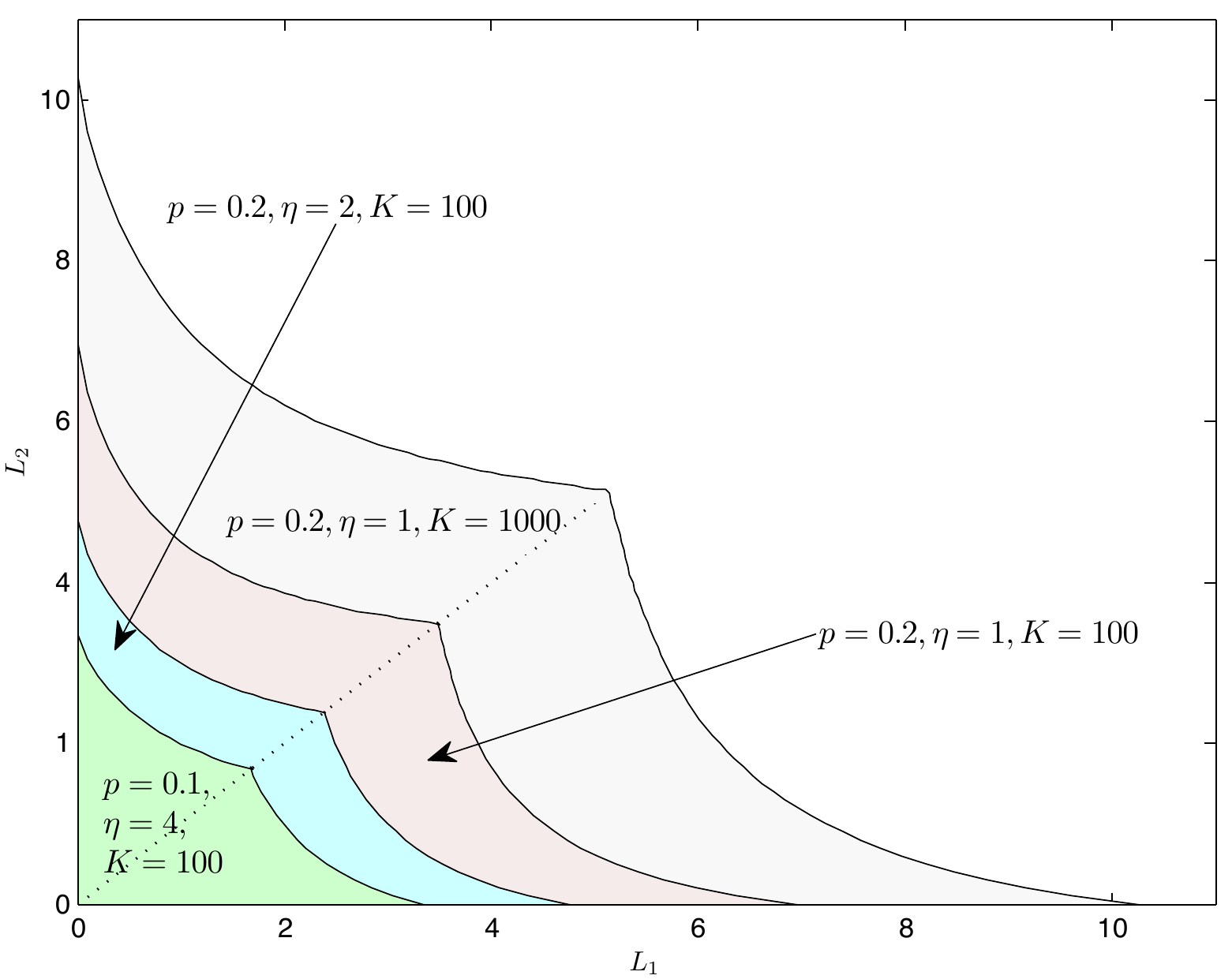}
	\caption{Achievable TR Region for the Two-Cells Scenario with Wyner Model, $D=2, \sigma_n^2=0.01, g=1$.}
	\label{fig:region_homo2}
\end{figure}

\begin{figure}[ht!]
\centering
	\includegraphics[width=2.7in]{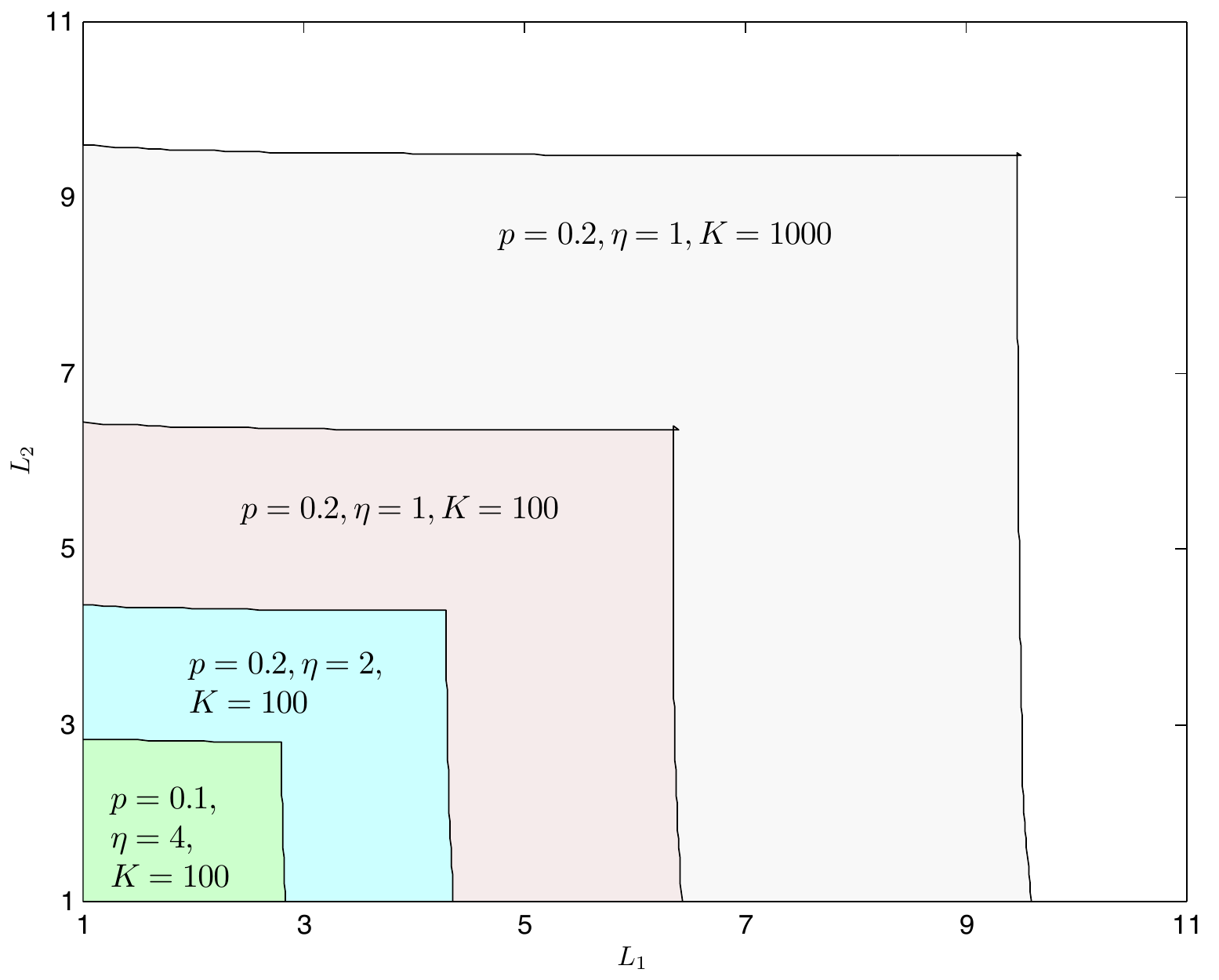}
	\caption{Achievable TR Region for the Two-Cells Scenario with Heterogeneous Users, $D=2, \alpha=3, \sigma_n^2=0.01$. }
	\label{fig:region_hetero}
\end{figure}

Let $\tilde{L}^*$ denote the maximum achievable TR with the relaxed constraints on $L$. Figure~\ref{fig:vsK} shows $\tilde{L}^*$ vs. $K$, for the single-cell scenario and two-cells scenario with equal TRs. As can be observed from the figure, for a fixed $K$, $\tilde{L}^*$ decreases as the cell size increases. This is because the users are uniformly located in the cell and as the cell size increases, the users' locations will be more spread out. As a consequence, the SINR on each beam will decrease and we must compensate this by decreasing the TR (to decrease the interferences). Figure~\ref{fig:vsSNR} shows $\tilde{L}^*$ vs. SNR for fixed number of users, where SNR is defined as $\frac{1}{\sigma_n^2}$. As can be observed from the figure, $\tilde{L}^*$ increases with SNR. Intuitively, when $\sigma_n^2$ decreases, we can increase the TR (effectively introduces additional interferences) while still satisfying the QoS constraints. Therefore, the achievable TR region will expand with decreased cell size or increased SNR. Finally, $\tilde{L}^*$ decreases as $M$, the number of cells, increases. This is because the SINR on each beam decreases with $M$.  

\begin{figure}[ht!]
\centering
		\includegraphics[width=3in]{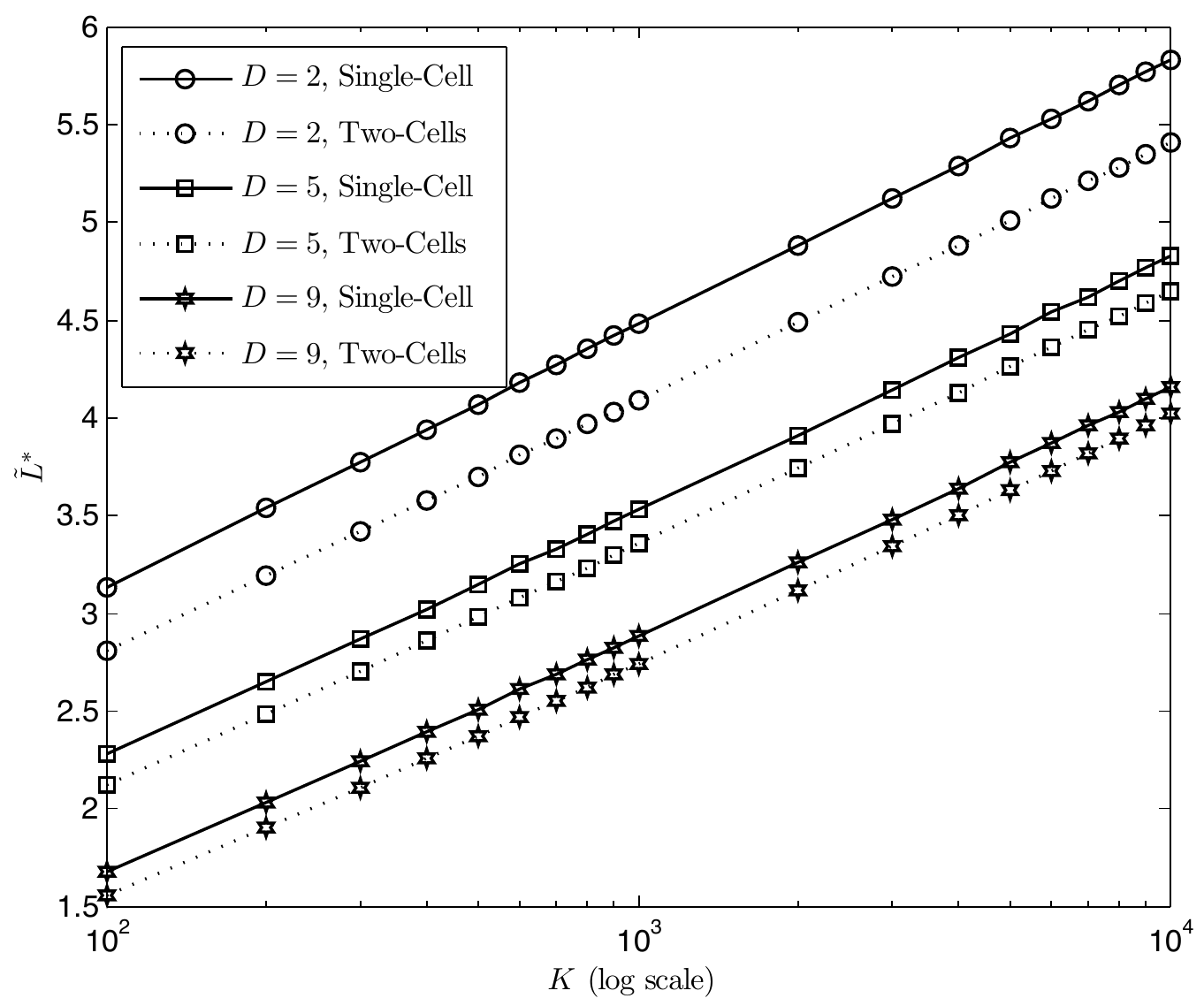}
	\caption{Maximum Achievable TR vs. $K$ for Single Cell and Two-Cell Scenarios with $p=0.1, \eta=4, \alpha=3$.}
	\label{fig:vsK}
\end{figure}

\section{Conclusions}\label{sec:conclusion}
In this paper, we considered a multi-cell multi-user MISO broadcast channel. Each cell employs the OBF scheme with variable TRs. We focused on finding the achievable TRs for the BSs to employ with a set of QoS constraints that ensures a guaranteed minimum rate per beam with a certain probability at each BS. We formulated this into a feasibility problem for the single-cell and multi-cell scenarios consisting of homogeneous users and heterogeneous users. Analytical expressions of the achievable TRs were derived for systems consisting of  homogeneous users and for systems consisting of heterogeneous users, expressions were derived which can be easily used to find the achievable TRs. An achievable TR region was obtained, which consists of all the achievable TR tuples for all the cells to satisfy the QoS constraints. Numerical results showed that the achievable TR region expands when the QoS constraints are relaxed, the SNR and the number of users in a cell are increased, and the size of the cells are decreased. 

\begin{figure}[ht!]
\centering
		\includegraphics[width=3in]{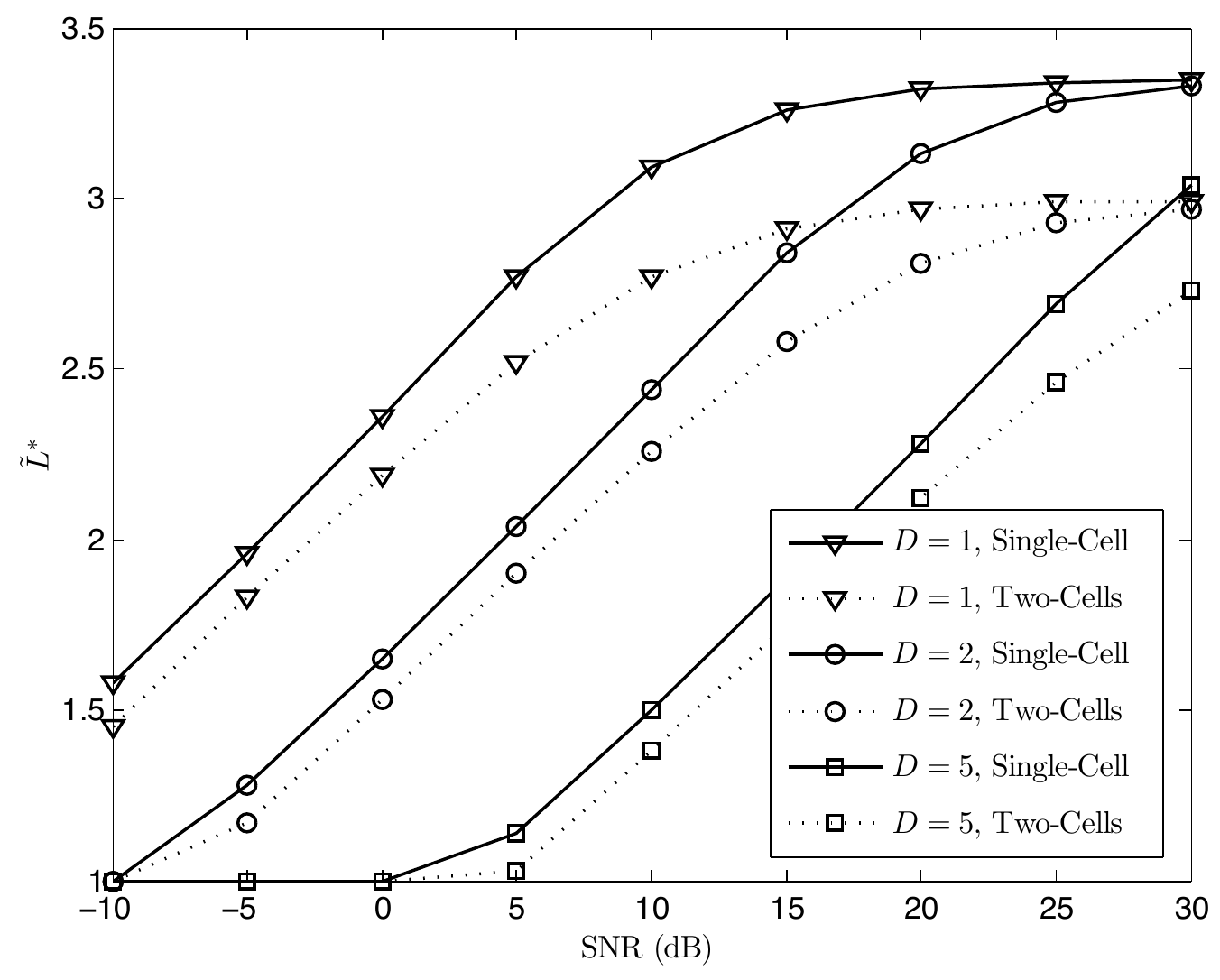}
	\caption{Maximum Achievable TR vs. SNR for Single Cell and Two-Cell Scenarios with $p=0.1, \eta=4, \alpha=3$.}
	\label{fig:vsSNR}
\end{figure}

\end{document}